\documentclass[12pt]{article}

\usepackage{amsmath,verbatim,amssymb,epsfig}

\usepackage{graphics,graphicx}
\usepackage{longtable}
\usepackage{amsmath}
\usepackage{bm}
\usepackage{placeins}

\usepackage{amsmath,natbib}
\usepackage[dvipsnames,usenames]{color}
\usepackage{epstopdf}

\newtheorem{proposition}{Proposition}

\newtheorem{corollary}{Corollary}
\newtheorem{defin}{\sc Definition}
\newenvironment{definition}{\begin{defin}\rm}{\end{defin}}
\newenvironment{proof}{\noindent{\sc Proof.}}{$\diamond$}

\def\be{\mbox{Be}}
\def\bin{\mbox{Bin}}

\def\dir{\mbox{Dir}}
\def\mult{\mbox{Mult}}

\def\E{\mbox{E}}
\def\V{\mbox{Var}}
\def\Cov{\mbox{Cov}}

\def\Cv{\mbox{Cov}}

\def\d{\mbox{d}}

\def\bp{{\bf p}}

\def\bw{{\bf w}}

\newcommand{\KL}{\mathrm{KL}}
\newcommand{\bbeta}{\boldsymbol{\beta}}

\newcommand{\Ree}{{\rm I}\!{\rm R}}

\newcommand{\DP}{\mathcal{DP}}

\newcommand{\AC}{\mathcal{A}}
\newcommand{\BB}{\mathcal{B}}

\newcommand{\YY}{\mathcal{Y}}

\newcommand{\PT}{\mathcal{PT}}

\begin{document}

\title{{\bf Characterising variation of nonparametric random probability measures using the Kullback-Leibler divergence}}
\author{
  {\sc  James Watson$^{1}$}, 
  {\sc  Luis Nieto-Barajas$^{1,2}$} and
  {\sc  Chris Holmes$^{1}$} \\
  {\sl {\small $^1$Department of Statistics, University of Oxford, UK}} \\ 
  {\sl {\small $^2$Department of Statistics, ITAM, Mexico}}}
\date{}

\maketitle

\begin{abstract}
This work studies the variation in Kullback-Leibler divergence between random draws from some popular nonparametric processes and their baseline measure. In particular we focus on the Dirichlet process, the P\'olya tree and the frequentist and Bayesian bootstrap. The results shed light on the support of these nonparametric processes. Of particular note are results for finite P\'olya trees that are used to model continuous random probability measures. Our results provide guidance for specifying the parameterisation of the P\'olya tree process that allows for greater understanding while highlighting limitations of the standard canonical choice of parameter settings.\end{abstract}

\section{Introduction}
\label{sec:intro}

Random probability models are key components of Bayesian nonparametrics \citep{hjort&al:10,ghosh:03,muller:04}  used to express prior beliefs with wide support. Bayesian nonparametrics has become increasingly popular in recent years due to the flexible modelling structures it supports and alleviating concerns over the ``closed hypothesis space'' of Bayesian inference. The most commonly used processes are the Dirichlet process prior or generalizations of it, and the P\'olya tree prior (PT), which includes the Dirichlet process as a special case, some of the main references being: \cite{ferguson:73,lavine:92,hjort&al:10}. Both the Dirichlet and the P\'olya tree priors are of particular interest because of their analytical tractability and their conjugacy properties for inference problems.

The properties of these processes are usually given at the level of characterizing their mean and variance when defining the process around a particular centring distribution $F_0$. For instance, if we have a random distribution $F$ with a Dirichlet process law, denoted $F\sim\DP(\alpha,F_0)$, where $F_0$ is the centring distribution, then $\E(F)=F_0$, and $\alpha$ is a precision parameter that controls the dispersion of $F$ from $F_0$. Similarly, if $F$ is a random distribution with law governed by a P\'olya tree process, using notation from \citet{hanson:06}, $F\sim\PT(\alpha,\rho,F_0)$, where $\rho$ denotes the precision function, then selecting $\alpha$ and a partition structure $\Pi$ that defines the tree, the draws will be centred around $F_0$, and $\alpha$ is again the precision parameter. Moreover, the precision function $\rho(\cdot)$ controls the speed at which the variance of the branching probabilities that define the PT increase or decrease. \cite{lavine:92} recommends $\rho(m)=m^2$ as a ``sensible canonical choice'', which as been adopted as the standard choice in the vast majority of applications, see for example \cite{karabatsos:06,muliere:97,walker:99,hanson&johnson:02,walker&mallick:97}. In practical applications when using P\'olya trees in Bayesian inference for example, it is also necessary to truncate the tree at a certain level $M$. One consequence of our work allows better insight for both choosing the truncation level $M$ and for choosing the function parameter $\rho(\cdot)$.

More generally we consider the general question of how far a random draw $F$, is from a specific centring distribution $F_0$. We also ask ourselves whether it is possible to set the parameters of the model in order to sample distributions at a specific divergence from $F_0$. In this note we provide some guidance on how to answer these questions using the most common measure of divergence between densities, the Kullback-Leibler (KL) divergence \citep{kullback&leibler:51}. We concentrate on this divergence for its fundamental role played in information theory and Bayesian statistics \citep[e.g.][]{kullback:97,bernardo:94,cover&thomas:91}.

Section 2 introduces some notation and defines the P\'olya tree as the principal model considered. We also consider the Bayesian and frequentist bootstrap procedures in section 3. Section 3 presents several properties of the KL divergence, considering random draws of some random probability models. Section 4 concludes with a discussion on the implications of these findings.

\section{Notation}
\label{sec:notation}

The P\'olya tree will be our main object of interest, particularly as the Dirichlet process can be seen as a particular case of a P\'olya tree, see \cite{ferguson:74}. We define it as follows. 

The P\'olya tree relies on a binary partition tree of the sample space. For simplicity of exposition we consider $(\Ree,\BB)$ as our measurable space with $\Ree$ the real line and $\BB$ the Borel sigma algebra of subsets of $\Ree$. Using the notation in \cite{nieto&mueller:08}, the binary partition tree is denoted by $\Pi=\{B_{mj}: m\in\mathbb{N}, j=1,..,2^m\}$, where the index $m$ denotes the level in the tree and $j$ the location of the partitioning subset within the level. The sets at level 1 are denoted by 
$(B_{11},B_{12})$; the partitioning subsets of $B_{11}$ are $(B_{21},B_{22})$, and $B_{12} = B_{23} \cup B_{24}$, such that $(B_{21},B_{22},B_{23},B_{24})$ denote the sets at level 2. In general, at level $m$, the set $B_{mj}$ splits into two disjoint sets $(B_{m+1,2j-1},B_{m+1,2j})$, where $B_{m+1,2j-1}\cap B_{m+1,2j}=\emptyset$ and $B_{m+1,2j-1}\cup B_{m+1,2j}=B_{mj}$.

We associate random branching probabilities $Y_{mj}$ with every set $B_{mj}$. 
We will use $F$ to denote a cdf or a probability measure in-distinctively, and $f$ to denote a density. We define $Y_{m+1,2j-1} = F(B_{m+1,2j-1} \mid B_{mj})$, and $Y_{m+1,2j}=1-Y_{m+1,2j-1} = F(B_{m+1,2j} \mid B_{mj})$. We denote by $\YY=\{Y_{mj}\}$ the set of random branching probabilities associated with the elements of $\Pi$. 

\begin{definition}
\citep{lavine:92}. Let $\AC_m=\{\alpha_{mj},\, j=1,\ldots,2^m\}$ be non-negative real numbers, $m=1,2,\ldots,$ and let $\AC=\bigcup \AC_m$. A random probability measure $F$ on $(\Ree,\BB)$ is said to have a P\'olya tree prior with parameters $(\Pi,\AC)$, if for $m=1,2,\ldots$ there exist random variables $\YY_m=\{Y_{m,2j-1}\}$ for $j=1,\ldots,2^{m-1}$, such that the following hold: 
\begin{enumerate}
\item All the random variables in $\YY=\cup_m\{\YY_m\}$ are independent. 
\item For every $m=1,2,\ldots$ and every $j=1,\ldots,2^{m-1}$, $Y_{m,2j-1}\sim\be(\alpha_{m,2j-1},\alpha_{m,2j})$.
\item For every $m=1,2,\ldots$ and every $j=1,\ldots,2^m$ 
$$F(B_{mj})=\prod_{k=1}^m Y_{m-k+1,j_{m-k+1}^{(m,j)}},$$
where $j_{k-1}^{(m,j)}=\lceil j_k^{(m,j)}/2 \rceil$ is a recursive decreasing formula, whose initial value is $j_{m}^{(m,j)}=j$, that locates the set $B_{mj}$ with its ancestors upwards in the tree. $\lceil\cdot\rceil$ denotes the ceiling function, and $Y_{m,2j}=1-Y_{m,2j-1}$ for $j=1,\ldots,2^{m-1}$.
\end{enumerate}
\end{definition}

There are several ways of centring the process around a parametric probability measure $F_0$. The simplest and most used method \citep{hanson&johnson:02} consists of matching the partition with the dyadic quantiles of the desired centring measure and keeping $\alpha_{mj}$ constant within each level $m$. More explicitly, at each level $m$ we take  
\begin{equation}
\label{bmj}
B_{mj}=\left(F_0^{-1}\left(\frac{j-1}{2^m}\right),F_0^{-1}\left(\frac{j}{2^m}\right)\right],
\end{equation}
for $j=1,\ldots,2^m$, with $F_0^{-1}(0)=-\infty$ and $F_0^{-1}(1)=\infty$. If we further take $\alpha_{mj}=\alpha_m$ for $j=1,\ldots,2^m$ we get $\E\{F(B_{mj})\}=F_0(B_{mj})$. 

In particular, we take $\alpha_{mj}=\alpha\rho(m)$, so that the parameter $\alpha$ can be interpreted as a precision parameter of the P\'olya tree \citep{walker&mallick:97}, and the function $\rho$ controls the speed at which the variance of the branching probabilities moves down in the tree. According to \cite{ferguson:74}, $\rho(m)=1/2^m$ defines an a.s. discrete measure that coincides with the Dirichlet process \citep{ferguson:73}, and $\rho(m)=1$ defines a continuous singular measure. Moreover, if $\rho$ is such that $\sum_{m=1}^\infty\rho(m)^{-1}<\infty$ it guarantees that $F$ is absolutely continuous \citep{kraft:64}, e.g., $\rho(m)=m^2,m^3,2^m,4^m$.

In practice we need to stop partitioning the space at a finite level $M$ to define a finite tree process. At the lowest level $M$, we can spread the probability within each set $B_{Mj}$ according to $f_0$. In this case the random probability measure defined will have a density of the form
\begin{equation}
f(x)=\prod_{m=1}^M Y_{m,j_m^{(X)}}2^M f_0(x),
\end{equation}
for $X\in\Ree$, and with $j_m^{(X)}$ identifying the set at level $m$ that contains $X$. This maintains the condition $\E(f)=f_0$. We denote a finite P\'olya tree process as $\PT_M(\alpha,\rho,F_0)$. Taking $M\to\infty$ defines a draw from a P\'olya tree.

Let us consider a set of functions $\rho(m)$ of the following types: 
\begin{equation}
\label{eq:rho}
\rho_1(m)=1/2^m,\;\;\rho_2(m)=1,\;\;\rho_3(m)=m^{\delta},\;\;\rm{and}\;\;\rho_4(m)=\delta^m, 
\end{equation}
where $\delta>1$, to define discrete, singular and two absolutely continuous measures, respectively. 

To measure ``distance'' between probability distributions, we concentrate on the Kullback-Leibler divergence, which for densities $f$ and $g$ is defined as
\begin{equation}
\label{eq:kl}
\KL(f||g)=\E_f\left[\log\left\{\frac{f(x)}{g(x)}\right\}\right]=\int\log\left\{\frac{f(x)}{g(x)}\right\}f(x)\d x.
\end{equation}

\section{Properties}\label{sec:prop}

\subsection{P\'olya Trees}

If $F\sim\PT_M(\alpha,\rho,F_0)$ then it is not difficult to show that the KL between the centring distribution $F_0$ and a random draw $F$ is a random variable that does not depend on $F_0$, and is given by:
\begin{equation}
\label{eq:kl1}
\KL(f_0 || f)=-\sum_{m=1}^M\sum_{j=1}^{2^m}(\log Y_{mj})\frac{1}{2^m}-M\log 2.
\end{equation}
Since the KL divergence measure is asymmetric, we can reverse the role of $f$ and $f_0$. In this case the reverse KL divergence becomes:
\begin{equation}
\label{eq:kl2}
\KL(f || f_0)=\sum_{m=1}^M\sum_{j=1}^{2^m}(\log Y_{mj})\prod_{k=1}^{m}Y_{m-k+1,j_{m-k+1}^{(m,j)}}+M\log 2.
\end{equation}

We now present some results that characterize the first two moments of these divergences. 

\begin{proposition}
\label{prop:1}
Let $F\sim\PT_M(\alpha,\rho,F_0)$. Then the Kullback-Leibler divergence between $f_0$ and $f$, defined in \eqref{eq:kl1}, has mean and variance given by
$$\E\{\KL(f_0 || f)\}=\sum_{m=1}^M\left\{\psi_0(2\alpha\rho(m))-\psi_0(\alpha\rho(m))-\log 2\right\}$$
and 
$$\V\{\KL(f_0 || f)\}=\sum_{m=1}^M\frac{1}{2^m}\left\{\psi_1(\alpha\rho(m))-2\psi_1(2\alpha\rho(m))\right\},$$
where $\psi_0(\cdot)$ and $\psi_1(\cdot)$ denote the digamma and trigamma functions respectively\footnote{The digamma function is defined as the logarithmic derivative of the gamma function, i.e. $\psi_0(x) = \frac{\mathrm{d}}{\mathrm{d}x} \log \Gamma(x) = \frac{\Gamma'(x)}{\Gamma(x)}$. In similar fashion, the trigamma function is defined as the second derivative.}. 
\end{proposition}
\begin{proof}
The expected value follows by noting that the geometric mean of a beta random variable is $\E(\log Y_{mj})=\psi_0(2\alpha\rho(m))-\psi_0(\alpha\rho(m))$. For the variance, we use the fact that the random variables $Y_{mj}$ are independent across $m$, and for the same $m$, $Y_{mj}$ and $Y_{mk}$ are independent for $|k-j|>1$. Noting that $\V(\log Y_{mj})=\psi_1(\alpha\rho(m))-\psi_1(2\alpha\rho(m))$ and since $Y_{m,2j}=1-Y_{m,2j-1}$, for $j=1,\ldots,2^{m-1}$, with $\Cov\{\log Y_{m,2j-1},\log(1-Y_{m,2j})\}=-\psi_1(2\alpha\rho(m))$, the result follows. 
\end{proof}

We now concentrate on the limiting behaviour of the expected KL value as a function of the finite tree level $M$. For some cases of the function $\rho(\cdot)$ this limit is finite. This is given in the following corollary. 

\begin{corollary}\label{cor:KLbound}
Let $\mathcal{E}_M:=\E\{KL(f_0 || f)\}$, given in Proposition \ref{prop:1}, to make explicit the dependence on the maximum level $M$. 
For the families $\rho_3(m)$ and $\rho_4(m)$ in expression (\ref{eq:rho}), the limit of the expected KL divergence, as $M\rightarrow\infty$, is bounded respectively by:
\begin{equation}
 \lim_{M\rightarrow\infty} \mathcal{E}_M \leq \frac{1}{4\alpha} \zeta(\delta) + \frac{1}{\alpha^2} \zeta(\delta^2)  
\end{equation}
\begin{equation}
\lim_{M\rightarrow \infty} \mathcal{E}_M  \leq \frac{\alpha(\delta+1) + 4}{4\alpha^2 (\delta^2 - 1)} 
\end{equation}
where $\delta$ is defined as in (\ref{eq:rho}), and $\zeta(\delta)= \sum_{n=1}^{\infty} n^{-\delta}$, is the Riemann zeta function.
\end{corollary}

\begin{proof}
The digamma function can be expanded as: $\psi_0(x) = \log x - (1/2)x^{-1} - \mathcal{O}(x^{-2})$, from which these inequalities follow.
\end{proof}

Taking instead the reverse KL, we have the following properties. 

\begin{proposition}
Let $F\sim\PT_M(\alpha,\rho,F_0)$. Then the Kullback-Leibler divergence between $f$ and $f_0$, defined in \eqref{eq:kl2}, has mean and variance given by
$$\E\{\KL(f || f_0)\}=\sum_{m=1}^M\left\{\psi_0(\alpha\rho(m)+1)-\psi_0(2\alpha\rho(m)+1)+\log 2\right\}$$
and 
$$\V\{\KL(f || f_0)\}=A+B,$$
where
$$A=\sum_{m=1}^M\left[\left\{\prod_{k=1}^m\left(\frac{\alpha\rho(k)+1}{2\alpha\rho(k)+1}\right)\right\}\lambda_5(m)-\left(\frac{1}{2}\right)^m\lambda_2^2(m)\right],$$
$$B=\sum_{m=1}^M\left(\left(\frac{\alpha\rho(m)}{2\alpha\rho(m)+1}\right)\left\{\prod_{k=1}^{m-1}\left(\frac{\alpha\rho(k)+1}{2\alpha\rho(k)+1}\right)\right\}\lambda_6(m)-\left(\frac{1}{2}\right)^m\lambda_2^2(m)\right.\hspace{3cm}$$
$$+\sum_{j=1}^{m-1}\left[\left(\frac{\alpha\rho(j)}{2\alpha\rho(j)+1}\right)\left\{\prod_{k=1}^{j-1}\left(\frac{\alpha\rho(k)+1}{2\alpha\rho(k)+1}\right)\right\}\lambda_2^2(m)-\left(\frac{1}{2}\right)^j\lambda_2^2(m)\right]\hspace{0cm}$$
$$+2\left\{\prod_{k=1}^{m-1}\left(\frac{\alpha\rho(k)+1}{2\alpha\rho(k)+1}\right)\right\}\sum_{j=m+1}^{M}\left\{\left(\frac{\alpha\rho(m)+1}{2\alpha\rho(m)+1}\right)\lambda_3(m)\lambda_2(j)\right.\hspace{1.1cm}$$
$$\left.\left.+\left(\frac{\alpha\rho(m)}{2\alpha\rho(m)+1}\right)\lambda_4(m)\lambda_2(j)-\lambda_2(m)\lambda_2(j)\right\}\right)\hspace{-7.8cm},$$
with 
$$\lambda_2(m)=\psi_0(\alpha\rho(m)+1)-\psi_0(2\alpha\rho(m)+1),$$
$$\lambda_3(m)=\psi_0(\alpha\rho(m)+2)-\psi_0(2\alpha\rho(m)+2),$$
$$\lambda_4(m)=\psi_0(\alpha\rho(m)+1)-\psi_0(2\alpha\rho(m)+2),$$
$$\lambda_5(m)=\psi_1(\alpha\rho(m)+2)-\psi_1(2\alpha\rho(m)+2)+\left\{\psi_0(\alpha\rho(m)+2)-\psi_0(2\alpha\rho(m)+2)\right\}^2,$$
$$\lambda_6(m)=\left\{\psi_0(\alpha\rho(m)+1)-\psi_0(2\alpha\rho(m)+2)\right\}^2-\psi_1(2\alpha\rho(m)+2).$$
\end{proposition}
\begin{proof}
The expected value follows by using independence properties and by noting that $\E\{(\log Y_{mj})Y_{mj}\}=\lambda_2(m)/2$. For the variance, we first bring the variance operator within the sum by splitting it into the sum of variances of each element plus the sum of covariances\footnote{The variance of each element is defined in terms of first and second moments and rely on independence properties to compute them. Working out the algebra with patience and noting that $\E\{(\log Y_{mj})Y_{mj}\}=\lambda_2(m)/2$, $\E\{(\log Y_{mj})Y_{mj}^2\}=\frac{1}{2}\left(\frac{\alpha\rho(m)+1}{2\alpha\rho(m)+1}\right)\lambda_3(m)$, $\E\{(\log Y_{mj})Y_{mj}(1-Y_{mj})\}=\frac{1}{2}\left(\frac{\alpha\rho(m)}{2\alpha\rho(m)+1}\right)\lambda_4(m)$, $\E\{(\log Y_{mj})^2Y_{mj}^2\}=\frac{1}{2}\left(\frac{\alpha\rho(m)+1}{2\alpha\rho(m)+1}\right)\lambda_5(m)$, and $\E\{(\log Y_{mj})\log(1-Y_{mj})Y_{mj}(1-Y_{mj})\}=\frac{1}{2}\left(\frac{\alpha\rho(m)}{2\alpha\rho(m)+1}\right)\lambda_6(m)$, the result is obtained.}.
\end{proof}

Figures \ref{fig:klmeans} and \ref{fig:klsd} respectively illustrate the behaviour of the mean and standard deviation, as a function of the truncation level $M$ for the two KL measures \eqref{eq:kl1} (empty dots) and \eqref{eq:kl2} (solid dots). The four panels in each figure correspond to choices of $\rho(m) = 1/2^m, 1, m^{\delta}, \delta^m$, as given in \eqref{eq:rho}. In all cases we use $\alpha=1$, and $\delta=2$ (the so-called canonical choice). The plots show that $\E\{\KL(f_0||f\}\geq\E\{KL(f || f_0)\}$ for all $M$ and for all functions $\rho$. Apart from the singular continuous case, $\rho_2(m) =1$, the variances of $\KL(f_0 || f)$ are also larger that those of $\KL(f || f_0)$. 

We see that for the case of $\rho_1(m) = 1/2^m$, which corresponds to the Dirichlet process, the mean value of the KL and the reverse KL diverge to infinity as $M\to\infty$\footnote{Figure \ref{fig:klmeans} appears to show that $\E\{\KL(f||f_0)\}$ remains constant, but this is an artefact due to the scale.}. The KL (\ref{eq:kl1}) increases at an exponential rate whereas for the reverse KL (\ref{eq:kl2}) the growth rate is constant. As for the standard deviations, that of the KL also diverges as $M\to\infty$, however, that of the reverse KL converges. 

The precision function $\rho_2(m)=1$, which defines a singular continuous random distribution \citep{ferguson:74}, has asymptotic constant expected values for both KL and reverse KL in the limit of $M$. The variance of the KL converges to a finite value when $M\to\infty$, but for the reverse KL the variance increases at a constant rate as a function of $M$. 
In the case of the two continuous processes, obtained with precision functions $\rho_3$ and $\rho_4$, the expected values for KL and the reverse KL converge in the limit, as given by the upper bounds in Corollary \ref{cor:KLbound}. Interestingly, the variances for the two KL divergences are asymptotically constant. 

These results give a precise interpretation of any choice of parametrisation of a P\'olya tree, summarised in the choice to the two parameters $(\alpha,\delta$)\footnote{Here we only consider the class of functions $\rho_3(m)=m^{\delta}$.}. 
The conventional method for informing the parametrisation of a P\'olya tree process is that the choice of $\delta$ is unimportant, with default value 2, and thus the choice of $\alpha$ completely controls the variety of draws from the process. However, these two parameters are confounded and should not be chosen independently.  
As shown in figure \ref{fig:KLcontour}, the expected KL is dependent on both parameters, although choices of $\alpha$ near to zero mean that the exponent $\delta$ has little effect on the expected KL and its variance.

\subsection{Frequentist and Bayesian ``bootstrap''}

Let us now consider the setting where $f_0$ is a discrete density with $n$ atoms $\{\xi_1,\ldots,\xi_n\}$, i.e., 
$f_0(x) = \sum_{i=1}^n p_i \delta_{\xi_i}(x)$, with $p_i > 0$ for all $i=1,\ldots,n$ and $\sum_{i=1}^n p_i=1$. Let $\bw=(w_1,\ldots,w_n)$ be random weights such that $w_i \geq 0$ and $\sum_{i=1}^n w_i =1$ almost surely. Let $f$ be a random distribution defined as a reweighing of the atoms of $f_0$ with the random weights $\bw$. In notation, $f(x)= \sum_{i=1}^n w_i\delta_{\xi_i}(x)$. 

The Kullback-Leibler divergence between $f_0$ and $f$ does not depend on the atoms locations and is given by:
\begin{equation}
\label{eq:kld1}
\KL(f_0 || f) = \sum_{i=1}^n p_i \log \left(\frac{p_i}{w_i}\right),
\end{equation}
and the reverse Kullback-Leibler has the form
\begin{equation}
\label{eq:kld2}
\KL(f || f_0) = \sum_{i=1}^n w_i\log\left(\frac{w_i}{p_i}\right).
\end{equation}

If we take $f_0$ to have uniform weights, such as when $X$ is a random sample from some population and $F_0$ represents the empirical CDF, we first highlight an important property in the relationship between the divergences \eqref{eq:kld1} and \eqref{eq:kld2}. 

\begin{proposition}
\label{prop:kl>12}
Consider the KL divergences \eqref{eq:kld1} and \eqref{eq:kld2}. If $p_i=1/n$ for $i=1,\ldots,n$, then for any given re-weighing vector $\bw$ taken from the simplex $\mathcal{Q}_n : =\{w : w_i\geq 0, \sum_{i=1}^n w_i = 1\}$ we have that 
\[\KL(f_0 || f)\geq KL(f || f_0).\]
\end{proposition}
\begin{proof}
Let $h(\bw):=\KL(f_0 || f)-\KL(f || f_0)$. Using expressions \eqref{eq:kld1} and \eqref{eq:kld2}, $h(w)$ becomes $h(\bw)=-\sum(1/n+w_i)\log(w_i)$. 
We note that $h(\bw)=0$ at $\bw^*=(1/n,\ldots,1/n)$ and is infinite on all the simplex boundaries. Moreover, $h$ is convex and by straightforward differentiation we see that $h''(\bw^*)$ is positive. The result follows.
\end{proof}

This result is consistent with the results from previous section. However, in this particular discrete setting $\KL(f_0 || f)$ dominates $\KL(f || f_0)$. 

Taking for instance $n\bw\sim\mult(n,\bp)$\footnote{We use this notation to emphasise the fact that $\bw$ represents a random probability mass function, but taking values on the set $\{0,1/n,2/n,..,1\}$. A factor of $n$ is needed for the vector to be distributed according to a multinomial distribution.}, a multinomial distribution with $n$ trials and $n$ categories with probability of success $\bp=(p_1,\ldots,p_n)$, means that the random $f$'s will be centred at $f_0$. It is not difficult to show that $\E(f)=f_0$. Note that if $p_i=1/n$ for $i=1,\ldots,n$ this choice of distribution for the weights $\bw$ coincides with the frequentist bootstrap \citep{efron:79} for which the atoms $\{\xi_i\}$ are replaced by i.i.d. random variables $\{X_i\}$. 

We note that the KL divergence \eqref{eq:kld1} will not in general be defined, as $w_i$ can be zero. In fact, for large $n$ and for $p_i=1/n$ in the previous multinomial choice, approximately one third of the weights will be zero. However, $0\log 0$ is defined by convention as 0, so the reverse KL \eqref{eq:kld2} is well defined. 

\begin{proposition} 
\label{prop:fb}
The expected value of the Kullback-Leibler between a ``bootstrap'' draw $f$, with $n\bw\sim\mult(n,\bp)$, and its centring distribution $f_0$, defined in \eqref{eq:kld2}, has the following upper bound:
\begin{equation}\label{prop:freqUpperBound}
\E\{\KL(f || f_0)\} \leq \sum_{i=1}^n  p_i\log\left(p_i +\frac{1-p_i}{n}\right)- H(\bp) 
\end{equation}
where $H(\bp) = \sum_{i=1}^n p_i \log p_i$, the entropy of the vector $\bp$.
For the special case when $p_i=1/n$, we have $\E\{KL(f || f_0)\} \leq \log\left(2-1/n\right) \leq \log 2$
\end{proposition}
\begin{proof}
$$\E\{\KL(f || f_0)\} = \sum_{i=1}^n \E\left\{w_i\log w_i\right\} - \sum_{i=1}^n\E\{w_i\}\log p_i.$$
Working on the individual expected values, 
\[ \E(w_i\log w_i ) = \sum _{k=1}^n {n \choose k}p_i^k (1-p_i)^{n-k} \left(\frac{k}{n}\right)\log\left(\frac{ k}{n}\right). \]
From which we get 
$\E(w_i\log w_i ) = (1/n)\E\{\log\left({(v_i+1)}/{n}\right)\}$, with $v_i\sim\bin(n-1,1/n)$.
Using Jensen's inequality we get
$\E\{w_i\log w_i \} \leq  p_i\log\left(p_i +{(1-p_i)}/{n}\right).$
Substituting this into the original sum and using $\E\{w_i\} = p_i$ gives the result.
\end{proof}

An alternative way of making the random $f$'s to be centred around $f_0$ is by sampling weights $\bw$ from a Dirichlet distribution with parameter vector $\bbeta=(\beta_1,\ldots,\beta_n)$ such that $\beta_i=\alpha_n p_i$, $i=1,\ldots,n$, with $\alpha_n>0$ a parameter changing as a function of the number of atoms. This is denoted $\bw\sim\dir(\alpha_n\bp)$.
It is straightforward to prove that $\E(f)=f_0$, and that the form of $\alpha_n$ parametrises the precision, analogous to the P\'olya tree case. 
If we take $\alpha_n = n$, $p_i=1/n$ and replace the atoms $\{\xi_i\}$ by i.i.d. random variables $\{X_i\}$, we obtain the original Bayesian bootstrap proposed by \cite{rubin:81}\footnote{It is interesting to note that in the original work they only consider this special case.}. Sampling from a Dirichlet with parameter vector $\alpha_n\bp$ gives a generalised version of this bootstrap procedure. \cite{ishwaran&zarepour:02} considered this model albeit in a different context. In this new setting, both the $\KL(f_0 || f)$ and the reverse $\KL(f || f_0)$, given in \eqref{eq:kld1} and \eqref{eq:kld2} respectively, are well defined since $w_i \neq 0$ almost surely. Their expected values and variances can be obtained in closed form as functions of $\alpha_n$ and $\bp$.

\begin{proposition}
Let $f$ be a ``generalised Bayesian bootstrap'' draw around $f_0$ with weights $\bw\sim\dir(\alpha_n\bp)$. Then the Kullback-Leibler divergence given in \eqref{eq:kld1} has mean and variance:
\[ \E\{\KL(f_0 || f)\} = H(\bp) - \sum_{i=1}^n p_i \left\{ \psi_0(\alpha_n p_i )  -  \psi_0(\alpha_n) \right\} \]
\[  Var\{\KL(f_0 || f)\}  =  \sum_{i=1}^n p_i ^ 2 \psi_1(\alpha_n p_i) - \psi_1(\alpha_n)\]
where $\psi_0$ and $\psi_1$ are the digamma and trigamma functions.
\end{proposition}

\begin{proof}
This result follows from $\E(\log w_i) = \psi_0(\alpha_n p_i) - \psi_0(\alpha_n)$ and linearity of expectation. The variance follows from $\V(\log w_i) = \psi_1(\alpha_n p_i ) - \psi_1(\alpha_n)$, and $\Cv(\log w_i,\log w_j) = \psi_1(\alpha_n p_i)\delta_{ij} - \psi_1(\alpha_n)$, where $\delta_{ij}$ is the Kronecker delta function taking value $1$ when $i=j$ and 0 otherwise.
\end{proof}

The limiting behaviour of this expected KL and its variance, as $n$ tends to infinity, can more easily be studied for the special case of $p_i=1/n$, $i=1,\ldots,n$. When $\alpha_n=\alpha$, i.e. constant, they both diverge to infinity. In the limit, this is a well known construction of a Dirichlet process, when the atoms are sampled i.i.d. from a baseline measure $G$. However, if we make $\alpha_n$ grow linearly with $n$, say $\alpha_n=\alpha n$, then $\lim_{n\to\infty} \E\{\KL(f_0 || f)\} =  \log(\alpha) - \psi_0(\alpha) $ and $ \lim_{n\rightarrow\infty} Var\{\KL(f_0 || f)\} = 0 $. These values are obtained by noting that $\psi_0(n)$ behaves like $\log(n)$ for large $n$. Finally, if we increase the rate at which $\alpha_n$ grows with $n$, say $\alpha_n=\alpha n^2$, both mean and variance of the KL converge to zero as $n\to\infty$.

\begin{proposition}
Let $f$ be a ``generalised Bayesian bootstrap'' draw around $f_0$ with weights $\bw\sim\dir(\alpha_n\bp)$. Then the Kullback-Leibler divergence given in \eqref{eq:kld2} has mean:
\begin{equation}
 \E \{ \KL(f || f_0) \}  = \sum_{i=1}^n  p_i\left\{ \psi_0(\alpha_n p_i + 1) - \psi_0(\alpha_n + 1)\right\} - H(\bp)
\end{equation}
where $H(\bp) := \sum_{i=1}^n p_i \log p_i$ the entropy of the vector $\bp$, and the variance given by
\begin{multline}\label{VarianceRevKL}
\V\left(\KL(f||f_0)\right) = \sum_{i=1}^n\left\{ \V(w_i \log w_i ) + (\log p_i)^2 \V(w_i) - 2 (\log p_i) \Cv(w_i\log w_i,w_i)\right\} \\
+2\sum_{i<j} \left\{ \Cv(w_i\log w_i,w_j\log w_j) + (\log p_i)(\log p_j) \Cv(w_i,w_j) - 2 (\log p_j) \Cv(w_i\log w_i, w_j ) \right\}
\end{multline}
where each of the elements are given in the footnote\footnote{ 
$\V(w_i ) = {p_i(1-p_i)}/{(\alpha_n+1)}$, $\Cv(w_i,w_j)=-p_i p_j/(\alpha_n+1)$, 
$ \V(w_i \log w_i) = {p_i (\alpha_n p_i +1)}/{(\alpha_n +1)} \{ \psi_1(\alpha_n p_i +2) - \psi_1(\alpha_n +2) + [\psi_0(\alpha_n p_i +2) -\psi_0(\alpha_n+2)]^2 \}
- p_i^2\{ \psi_0(\alpha_n p_i + 1) - \psi_0(\alpha_n+1) \}^2 $, 
$\Cv(w_i \log w_i, w_i) = {p_i(\alpha_n p_i +1)}/{(\alpha_n +1)}\{ \psi_0(\alpha_n p_i + 2) - \psi_0(\alpha_n+2) \} - p_i^2 \{ \psi_0(\alpha_n p_i + 1) - \psi_0(\alpha_n +1) \}$,
$\Cv(w_i \log w_i, w_j) = p_i p_j \{ -{\psi_0(\alpha_n p_i +1)}/{(\alpha_n +1)} + \psi_0(\alpha_n +1) - {\alpha_n \psi_0(\alpha_n + 2)}/{(\alpha_n +1)} \}$, 
$\Cv(w_i \log w_i, w_j \log w_j) = {\alpha_n p_i p_j}/{(\alpha_n + 1)} [\{\psi_0(\alpha_n p_i +1) - \psi_0(\alpha_n +2)\}\{\psi_0(\alpha_n p_j +1) - \psi_0(\alpha_n + 2)\}- \psi_1(\alpha_n + 2)  ] - p_i p_j \{\psi_0(\alpha_n p_i + 1) - \psi_0(\alpha_n + 1)\}\{\psi_0(\alpha_n p_j + 1) - \psi_0(\alpha_n + 1) \}$.}. 
\end{proposition}
\begin{proof}
Note that each $w_i \sim \be\{\alpha_n p_i, \alpha_n(1-p_i)\}$ and thus we have that $\E(w_i\log w_i) = p_i\{ \psi_0(\alpha_n p_i + 1) - \psi_0(\alpha_n + 1)\}$. 
Using linearity of expectation and substituting this expression we obtain the mean. Using properties of the variance and covariance of sums we get the second part of the result.
\end{proof}

Similarly to the previous case, if we take $p_i=1/n$ and $\alpha_n=\alpha n$, when $n\to\infty$ then $\E\{\KL(f||f_0)\}\to\psi_0(\alpha+1)-\log(\alpha)$. It is possible to show analytically that each term in (\ref{VarianceRevKL}) goes to zero as $n\rightarrow\infty$, but this can also be seen using the relation between the two KLs given in Proposition \ref{prop:kl>12}, and noting that the variance involves a monotonic transformation, hence we have that $\V\{\KL(f_0||f)\}\geq\V\{\KL(f||f_0)\}$. From the previous result it follows that $\lim_{n\to\infty}\V\{\KL(f||f_0)\}=0$  for these choices of $p_i$ and $\alpha_n$. 

In Figure \ref{fig:kld} we compare the expected value and variance of both KL and reverse KL for $p_i=1/n$ and different values of $\alpha_n$ as a function of $n$. The first column corresponds to $\alpha_n=1$, the second column to $\alpha_n=n$ and the third to $\alpha_n=n^2$, which induce high, moderate and small variance in the $\bw$ respectively. In accordance to what we have proved, the expected value and variance of $\KL(f_0||f)$ are larger than those of $\KL(f||f_0)$, and their limiting behaviours can also be assessed from the graphs. 

If we replace the atoms $\{\xi_i\}$ by i.i.d. random variables $\{X_i\}$ from a distribution $G$ and take $p_i=1/n$ for $i=1,\ldots,n$, then $f_0$ represents the empirical density for the random variables $X_i$'s and $f$ represents a random process centred around the empirical. \cite{ishwaran&zarepour:02} considered exactly this random probability process and derived results for the limiting behaviour for a variety of choices of $\alpha_n$ (see Theorem 3, page 948). Let $F$ be the cdf associated to $f$. When $\alpha_n=\alpha$, then $F$ is distributed according to a Dirichlet process $\DP(\alpha,G)$, in the limit as $n\to\infty$. If $\alpha_n=\alpha n$, then we have almost sure weak convergence of $F$ to $G$, as $n\to\infty$. For the third case considered here, $\alpha_n=\alpha n^2$, $F$ converges in probability to $G$, as $n\to\infty$. 

The case where $\alpha_n=\alpha n$ is of particular interest. Although we have weak convergence of $F\to G$, the random distribution does not converge in KL divergence. In other words, although functionals of $f$ tend to the functionals of $f_0$, the KL divergence between the two densities remains non zero. This becomes apparent when considering the random quantity $n w_i$, which comes into the equation (\ref{eq:kld1}), whose variance becomes asymptotically $1/\alpha$, as $n\to\infty$.
Convergence in Kullback-Leibler is a strong statement, stronger than convergence of functionals and $L_1$ convergence. A more intuitive illustration is the posterior convergence of two Dirichlet processes with different baseline measures (that have the same support). By posterior consistency, both will weakly converge to the same measure, but their $L_1$ divergence will remain finite and their KL divergence will remain infinite. 

\section{Discussion}

This note explores properties of the KL and reverse KL of draws $F$ from some classical random probability models with respect to their centring distribution $F_0$. These properties become relevant when applying a particular process as a modelling tool. For example, draws from the Dirichlet process prior have divergent expected KL (obtained in our P\'olya tree setting with $\rho_1$ in \eqref{eq:rho} and $M\to\infty$, and also obtained in the Bayesian bootstrap setting with $\alpha_n=\alpha$ and $n\to\infty$). Therefore we can say that any random draw taken from a Dirichlet process prior is completely ``different'' from the baseline distribution as measured in terms of the KL divergence, regardless the value of $\alpha$. This is also a surprising result but in accordance with the full support property of the Dirichlet process\footnote{Stated in \citet{ferguson:73}, saying that any fixed density (measure) $g$ absolutely continuous with respect to $f_0$ can be arbitrarily approximated with a draw $f$ from a Dirichlet process.}.

Our key result concerns the P\'olya tree prior. 
In the majority of applications, it is usually constructed in its continuous version, i.e. the precision function $\rho$ satisfies the continuity property, for example $\rho_3$ and $\rho_4$ as given in \eqref{eq:rho}. In these cases, the first two moments of the distance (in KL units) of the draws from their centring measure is given as an explicit function of the truncation level $M$ and the precision function $\alpha\rho(m)$.
Therefore the specification of the truncation level $M$, precision parameter $\alpha$, and precision function $\rho$ are all highly important, with careless choices leading to a prior overly concentrated around $f_0$. The vast majority of applications with P\'olya tree priors use the family $\rho_3(m)$ with choice of exponent $\delta=2$.
In Figure \ref{fig:klar} we show that using $\rho_3(m)$ with a choice of $\delta=1.01$ (empty dots) gives greater gains in expected KL as $M$ is increased as compared to those obtained for the standard choice of $\delta=2$ and decreasing the parameter $\alpha$. The concentration around the baseline measure is highly sensitive to this choice of exponent, thus questioning the ``sensible canonical choice'' of $\delta=2$ given by \cite{lavine:92}.

Moreover, in practice P\'olya trees are used in their finite versions, that is, finite $M$. In such cases the choice of $M$ has been done with a rule of thumb \citep[e.g.][]{hanson:06}, say $M=\log_2(n)$ with $n$ being the data sample size. The authors note 'a law of diminishing returns' when increasing the truncation level from $M\to M+1$. Our study confirms this by plotting the diversity of draws as measured in KL against $M$, and these findings suggest that a P\'olya tree prior with as a low as $M=4$ and $\rho(m)=2^m$ can produce random draws that are equally far from the centring distribution as with a larger $M$ (see two bottom panels in Figures \ref{fig:klmeans} and \ref{fig:klsd}). If it desired to make proper use of finite nature of the tree, the various possibilities in specification of the precision function $\rho$ within families that satisfy the continuity property should be used.

In the discrete setting, we can always see $f_0$ as the empirical density obtained from a sample of size $n$ taken from a continuous density. This is often the case when characterising a posterior distribution in Bayesian analysis, for example via MCMC sampling \citep[e.g.][]{gelman&al:02}. One lesson from this work, is that by increasing $n$, the variance of the reverse KL in the frequentist bootstrap, and the variance of the KL and reverse KL for the Bayesian bootstrap with $\alpha_n=\alpha n$, converge to zero. This implies that for large $n$ a frequentist or Bayesian bootstrap draw lies below $\log(2)$ and exactly at $\log(\alpha)-\psi_0(\alpha)$ or $\psi_0(\alpha+1)-\log(\alpha)$ in KL units, respectively. 

\section*{Acknowledgements}
We are grateful to Judith Rousseau for helpful comments. 
Watson is supported by the Industrial Doctoral Training Centre (SABS-IDC) at Oxford University and Hoffman-La Roche.  
This work was done whilst Nieto-Barajas was visiting the Department of Statistics at the University of Oxford. He is supported by {\it Asociaci\'on Mexicana de Cultura, A.C.}--Mexico. 
Holmes gratefully acknowledges support for this research from the EPSRC and the Medical Research Council.

\bibliographystyle{authordate1}
\bibliography{ref}

\begin{figure}
\centering
\includegraphics[scale=.5]{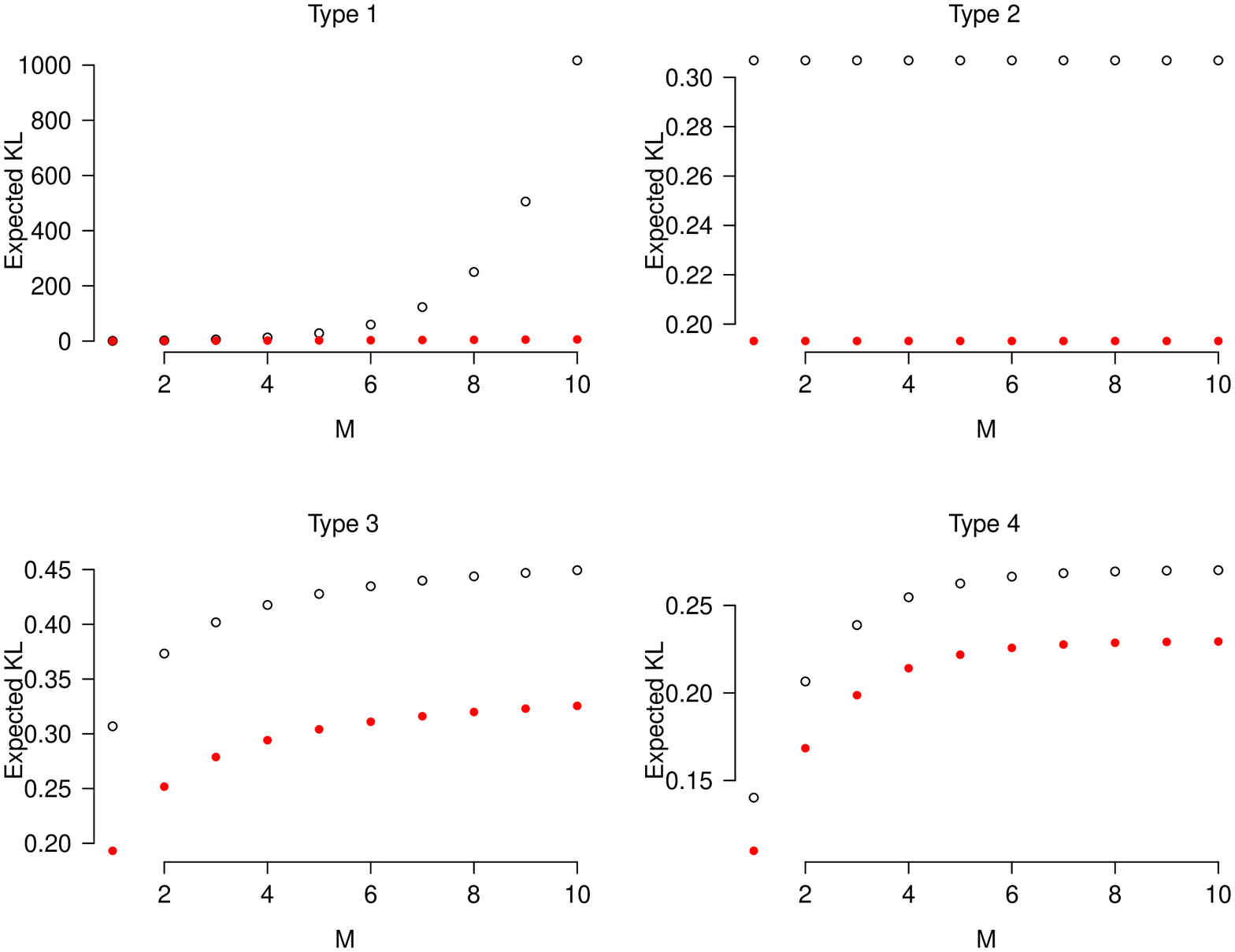}
\caption{Comparison between expected values of KL for different values of $M$. $\E\{\KL(f_0 || f)\}$ (empty dots) and $\E\{\KL(f || f_0)\}$ (solid dots). Type 1 to 4 denote the different $\rho$ functions as in \eqref{eq:rho}.} 
\label{fig:klmeans}
\end{figure}

\begin{figure}
\centering
\includegraphics[scale=.5]{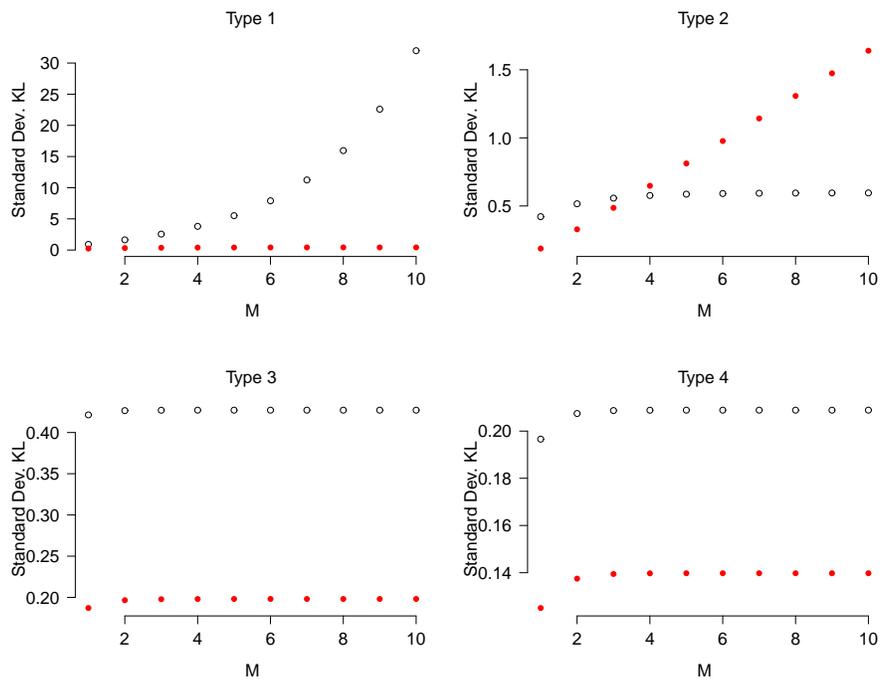}
\caption{Comparison between standard deviations of KL for different values of $M$. $\sqrt{\V\{KL(f_0 || f)\}}$ (empty dots) and $\sqrt{\V\{\KL(f || f_0)\}}$ (solid dots). Type 1 to 4 denote the different $\rho$ functions as in \eqref{eq:rho}.} 
\label{fig:klsd}
\end{figure}

\begin{figure}
\centering
\includegraphics[scale=.5]{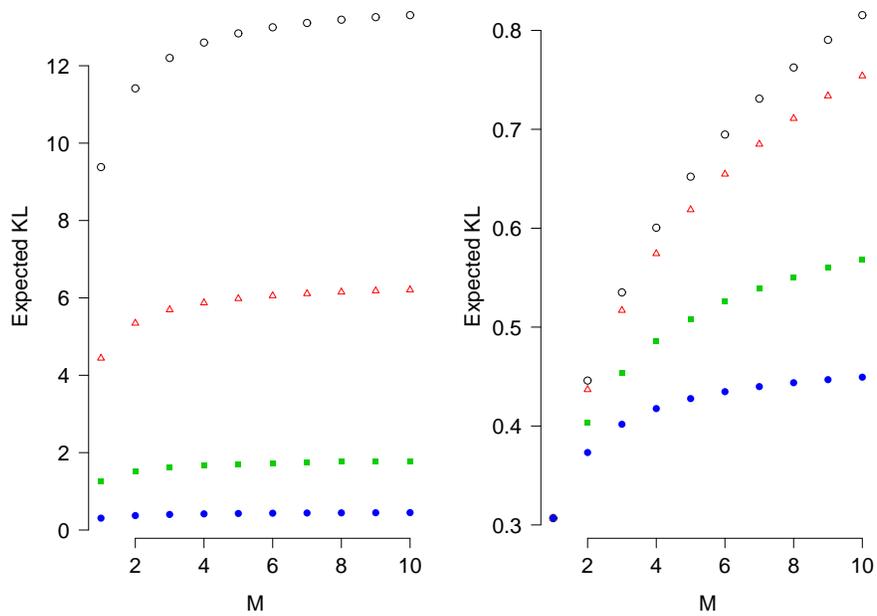}
\caption{Expected $\KL(f_0||f)$ for varying $\alpha$ (left panel) and $\rho(m)=m^{\delta}$ (right panel). Left: $\alpha=0.05$ (empty dots), $\alpha=0.1$ (triangles), $\alpha=0.3$ (squares), $\alpha=1$ (solid dots). Right: $\delta=1.01$ (empty dots), $\delta=1.1$ (triangles), $\delta=1.5$ (squares), $\delta=2$ (solid dots).} 
\label{fig:klar}
\end{figure}

\begin{figure}
\centering
\includegraphics[scale=.5]{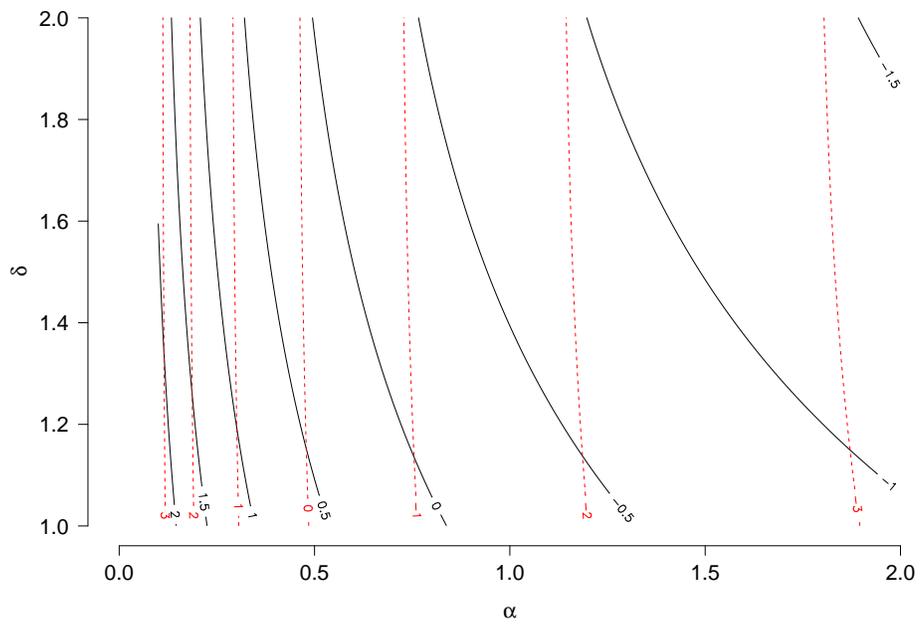}
\caption{Overlaid contour plots of the log-expected KL (black lines) and the log-variance of the KL (dashed red lines) as functions of the two parameters $(\alpha,\delta)$, at regular intervals of 1/2 and 1 respectively, of draws from a P\'olya tree process with truncation level $M=10$.}
\label{fig:KLcontour}
\end{figure}

\begin{figure}
\centering
\includegraphics[scale=.5]{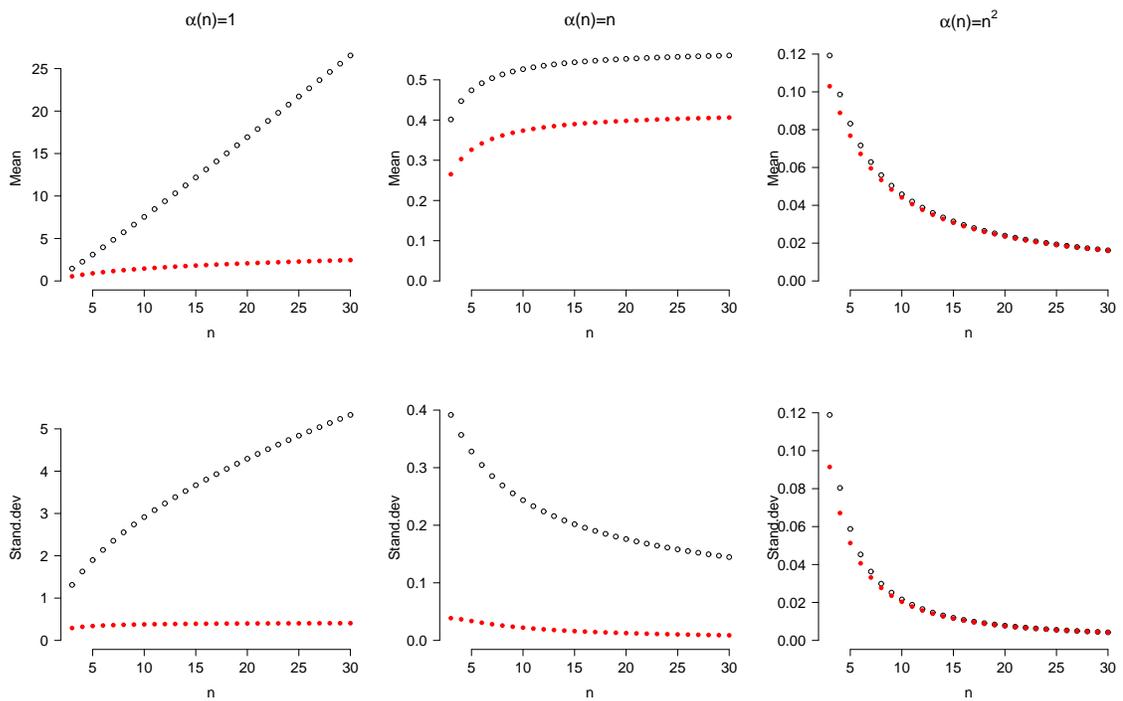}
\caption{Expected value (top row) and standard deviation (bottom row) of $KL(f_0||f)$ (black empty dots) and $\KL( f || f_0 )$ (red solid dots). In columns from left to right: $\alpha_n = 1$, $\alpha_n = n$ and $\alpha_n = n^2$.} 
\label{fig:kld}
\end{figure}

\end{document}